\documentclass[thmsa,12pt]{article}
\normalfont\sffamily
\usepackage{amsfonts}
\usepackage{amssymb}
\usepackage[active]{srcltx}
\usepackage{color}
\usepackage[x11names]{xcolor}
\usepackage{amsmath}
\usepackage{mathtools}
\usepackage{amsthm}
\usepackage{bm} 
\usepackage{graphicx}
\usepackage{textcomp}
\usepackage{a4wide}
\usepackage{multirow}
\usepackage{diagbox}
\usepackage{color}
\usepackage{makecell}
\usepackage{longtable}
\usepackage{cellspace}
\usepackage{tabularx}
\usepackage{accents}
\usepackage{array,booktabs,ragged2e}
\usepackage{pxfonts}
\usepackage{marvosym}
\usepackage{dsfont} 
\usepackage[format=hang]{caption}
\usepackage{subcaption}
\usepackage{url}
\usepackage{xfrac}

\usepackage[normalem]{ulem}
\usepackage{comment}
\usepackage{pxfonts}
\usepackage[margin=3cm]{geometry} 
\usepackage{titlesec}

\usepackage{setspace}
\usepackage{footmisc}
\usepackage{accents}
\usepackage{appendix}

\usepackage{array,booktabs,ragged2e}
\newcolumntype{R}[1]{>{\RaggedRight}p{#1}}

\usepackage{cellspace}
\usepackage{hyperref}
\newtheorem{theorem}{Theorem}
\newtheorem{lemma}{Lemma}
\newtheorem{example}{Example}

\newtheorem{definition}{Definition}

\def \be {\begin{equation}}
\def \ee {\end{equation}}

\makeatletter
\def\@seccntformat#1{\@ifundefined{#1@cntformat}%
   {\csname the#1\endcsname\quad}
   {\csname #1@cntformat\endcsname}}
\newcommand\section@cntformat{}     
\makeatother

\begin{document}

\setstretch{1.2}


\title{The Public Good index for games with several levels of approval in the input and output}
\author{Sascha Kurz\\ \footnotesize University of Bayreuth\\\footnotesize sascha.kurz@uni-bayreuth.de}
\date{}
\maketitle
\begin{abstract}
  The Public Good index is a power index for simple games introduced by Holler and later axiomatized by Holler and Packel, so that some authors also speak of the 
  Holler--Packel index.\footnote{The paper is dedicated to the occasion of the 75th birthday of Manfred J.\ Holler.} A generalization to the class of games 
  with transferable utility was given by Holler and Li. Here we generalize the underlying ideas to games with several levels of approval in the input and output -- 
  so-called $(j,k)$ simple games. Corresponding axiomatizations are also provided.
  
  \medskip
  
  \noindent
  \textbf{Keywords:} Public Good index, Public Good value, $(j,k)$ simple games, simple games, TU games, values, axiomatization
\end{abstract}

\subsection{Introduction}
\label{sec_intro}
Assume that you are submitting a paper to a computer science conference (or some other scientific discipline with a similar reviewing convention). You paper is usually 
send to several reviewers, which are typically chosen by the programming committee or assign themselves in some kind of bidding procedure. Unattached the selection process, 
for each paper there exists a set $N$ of reviewers. The task of the reviewers is to read and to evaluate the submitted paper. Besides some comments and remarks in free text, 
a summarizing evaluation according to a certain predefined scale is requested. A typical scale consists e.g.\ of the possible answers {\lq\lq}strong accept{\rq\rq}, 
{\lq\lq}accept{\rq\rq}, {\lq\lq}weak accept{\rq\rq}, {\lq\lq}borderline{\rq\rq}, {\lq\lq}weak reject{\rq\rq}, {\lq\lq}reject{\rq\rq}, and {\lq\lq}strong reject{\rq\rq}. 
After every reviewer has announced his or her evaluation, these individual opinions are summarized to a group decision, where we assume that only the outcomes 
{\lq\lq}accept{\rq\rq} or {\lq\lq}reject{\rq\rq} are possible. Of course, this oversimplifies the practical setting where we may have discussion rounds between the 
reviewers with the possibility to adjust their evaluations or some kind of interaction with the authors of the paper. Such a decision rule $v$ may be formalized as follows: 
For some set of agents $N$ and a set of levels of approval for the input $J$, each vector in $J^{|N|}$ is mapped to an element of the set of levels of approval in the output 
$K$. In our example we have $|J|=7$ and $|K|=2$, but may also consider an output set $K$ of cardinality three by distinguishing between a lecture, a poster presentation, or 
rejection. If the options in $J$ can be mapped to a numerical score, like e.g.\ $+3,+2,+1,0,-1-,2,-3$ in our example, then such a decision rule might be simply given by some 
threshold $\tau$. I.e., accept all papers with mean of the scores at least $\tau$. However, rules might be more complicated including extra conditions like e.g.\ requiring 
that no paper with at least one {\lq\lq}strong reject{\rq\rq} is accepted. Given a specific decision rule $v$ one might ask for the {\lq\lq}influence{\rq\rq} of a specific 
agent $i\in N$ on the group decision. Having only homogeneous agents in mind this question does not seem to make too much sense. However, agents may also be heterogeneous. 
In our example the reviewers may have different levels of expertise, which is indeed a common query to the reviewer when writing his or her evaluation. Of course, we as the 
author of the paper usually do not have the details to determine the influence of the individual reviewers and should have little interest to do so, but the author of the 
one day is the organizer of a huge conference the other day and possibly in charge to design the details of the decision rules. 

Taking our exemplifying story aside, we can clearly imagine situations where the individual opinions of $|N|$ agents from an ordered set $J$ of inputs are mapped 
to an output from an ordered set $K$. To this end $(|J|,|K|)$ simple games have been introduced, see e.g.\ \cite{freixas2003weighted,freixas2009anonymous}, and we remark 
that simple games are in one-to-one correspondence to $(2,2)$ simple games with $J=\{0,1\}$ and $K=\{0,1\}$. Measurements of influence for simple games are also called 
power indices and the Public Good index, introduced in \cite{holler1982forming}, is a particular example. The question of this paper is whether a measure in the vein of 
the Public Good index can be defined for the class of $(j,k)$ simple games. We motivate a few variants and give a axiomatizations. An axiomatization of the Public Good 
index for simple games was given in \cite{holler1983power}, so that some people also speak of the Holler--Packel index, and the generalization to TU games was axiomatized in 
\cite{holler1995public}. A different axiomatization, for both cases and based on potential functions, was given in \cite{haradau2007potential}. For $(j,2)$ simple games 
a Public Good index was recently introduced in \cite{sebastien2020public} along with two axiomatizations.

The remaining part of this paper is structured as follows. In Section~\ref{sec_preliminaries} we summarize some necessary preliminaries from the literature before we discuss 
different generalizations of the Public Good index and corresponding axiomatizations to the class of $(j,k)$ simple games in Section~\ref{sec_axiom}.

\subsection{Preliminaries}
\label{sec_preliminaries}

Let $N=\left\{ 1,2,...,n\right\}$ be a finite set of agents or voters. Any subset $S$ of $N$ is called a coalition and the set of all coalitions of $N$ is denoted by
the power set $2^{N}$. For given integers $j,k\ge 2$ we denote by $J=\{0,\dots,j-1\}$ the possible input levels and by $K=\{0,\dots, k-1\}$ the possible
output levels, respectively. We write $x\le y$ for $x,y\in\mathbb{R}^n$ if
$x_i\le y_i$ for all $1\le i\le n$. For each $\emptyset\subseteq S\subseteq N$ we write $x_S$ for the restriction
of $x\in\mathbb{R}^n$ to $\left(x_i\right)_{i\in S}$. As an abbreviation, we write $x_{-S}=x_{N\backslash S}$. Instead of $x_{\{i\}}$ and
$x_{-\{i\}}$ we write $x_i$ and $x_{-i}$, respectively. Slightly abusing notation we write $\mathbf{a}\in\mathbb{R}^n$, for the vector that
entirely consists of $a$'s, e.g., $\mathbf{0}$ for the all zero vector. 

\begin{definition}
  Let $j,k\ge 2$ and $n\ge 0$ be integers. A \emph{(j,k) simple game} is a mapping $v\colon J^n\to K$ satisfying $v(\mathbf{0})=0$ and $v(x)\le v(y)$ for all 
  $x,y\in J^n$ with $x\le y$.\footnote{Some authors also require $v(\mathbf{j-1})=k-1$, which would clash with the potential function approach as it is the case 
  for simple games. Note that we have reversed the order of the input levels of approval compared to \cite{freixas2003weighted}.} 
\end{definition}

\begin{example}
  \label{ex_3_3_simple_game}
  For $n=j=k=3$ let the $(3,3)$ simple game $v$ be defined via 
  $$
    v(x)=\left\{\begin{array}{rcl}
      0 &:& 3x_1+2x_2+x_3<7 \\ 
      1 &:& 7\le 3x_1+2x_2+x_3 <12\\ 
      2 &:& x_1=x_2=x_3=2     
    \end{array}\right.
  $$  
  for all $x\in \{0,1,2\}^3$.
\end{example}

\begin{definition}
  A \emph{simple game} is a mapping $v\colon 2^N\to\{0,1\}$ that satisfies $v(\emptyset)=0$, $v(N)=1$, and $v(S)\le v(T)$ for all $\emptyset\subseteq S\subseteq T\subseteq N$, where 
  the finite set $N$ is called the \emph{player set} or \emph{set of players}.\footnote{In some papers $v(S)\le v(T)$ is dropped in the definition of a simple game and they speak of \emph{monotonic 
  simple games} is it is additionally assumed. For the potential function approach we will drop the condition $v(N)=1$ later on, while it is indeed necessary for the normalized 
  Public Good index.}
\end{definition}
Let $v$ be a simple game with player set $N$. A subset $S\subseteq N$ is called \emph{winning coalition} if $v(S)=1$ and \emph{losing coalition} otherwise. A winning coalition 
$S\subseteq N$ is called \emph{minimal winning coalition} if all proper subsets $T\subsetneq S$ of $S$ are losing. The set of minimal winning coalitions is denoted by 
$\operatorname{MWC}(v)$.

\begin{example}
  \label{ex_simple_game}
  For player set $N=\{1,2,3\}$ let $v$ be the simple game defined by $v(S)=1$ iff $w(S):=\sum_{i\in S}w_i\ge 3$ and $v(S)=0$ otherwise for all 
  $S\subseteq N$, where $w_1=3$, $w_2=2$, and $w_3=1$.
\end{example}

The winning coalitions of the simple game from Example~\ref{ex_simple_game} are given by $\{1\}$, $\{2,3\}$, $\{1,2\}$, $\{1,3\}$, and $\{1,2,3\}$. Only $\{1\}$ and $\{2,3\}$ are 
minimal winning coalitions.

In order to embed a given simple game $v\colon 2^N\to\{0,1\}$ as a $(2,2)$ simple game $\hat{v}$ with $J=\{0,1\}$ and $K=\{0,1\}$, we assume $N=\{1,\dots,n\}$. To each 
coalition $S\subseteq N$ we assign the vector $x^S\in \{0,1\}^n$ with $x_i^S=1$ iff $i\in S$ and $x_i^S=0$ otherwise. Given a vector $x\in\{0,1\}^n$ the corresponding 
coalition is given by $S=\{i\in N\mid x_i=1\}$, so that $v(S)=\hat{v}(x^S)$.

The (raw) Public Good index for a simple game $v$ with player set $N$ and a player $i\in N$ is given by
\begin{equation}
  \operatorname{PGI}_i(v)=|\left\{S\in \operatorname{MWC}(v)\mid i\in S\right\}|.
\end{equation} 
With this, the (normalized) Public Good index is given by 
\begin{equation}
  \overline{\operatorname{PGI}_i}(v)=\frac{\operatorname{PGI}_i(v)}{\sum_{j\in N} \operatorname{PGI}_j(v) }
\end{equation}
and is e.g.\ \emph{efficient}, i.e., $\sum_{i\in N} \overline{\operatorname{PGI}_i}(v)=1$. Note that for the normalized version it is important to assume 
that $v(N)=1$ since $\operatorname{MWC}(v)$ is empty otherwise, so that $\overline{\operatorname{PGI}_i}(v)$ would be undefined.

A generalization of simple games, without the monotonicity assumption, are games with transferable utility -- so-called TU games. 
\begin{definition}
  A \emph{TU game} is a mapping $v\colon 2^N\to\mathbb{R}$ with $v(\emptyset)=0$, where the finite set $N$ is called the \emph{player set} or \emph{set of players}.  
\end{definition}
If we additionally assume $v(S)\le v(T)$ for all $\emptyset \subseteq S\subseteq T\subseteq N$, we speak of a \emph{monotone TU game} or a \emph{capacity}.

The analog of minimal winning coalitions in the context of TU games are \emph{minimal crucial coalitions}, see e.g.\ \cite{haradau2007potential} or 
\emph{real gaining coalitions}, see \cite{holler1995public}. To this end, we call a player $i\in S\subseteq N$ \emph{crucial} in a TU game $v$ if $v(S)>v(S\backslash i)$. 
A coalition $S$ in which every player $i$ is crucial is called \emph{minimal crucial coalition} and the set of minimal crucial coalitions is denoted by 
$\operatorname{MCC}(v)$. A coalition $S\subseteq N$ is called a real gaining coalition if $v(S)-v(T)>0$ for all proper subsets $\emptyset \subseteq T\subsetneq S$ of $S$. 
The set of all real gaining coalitions of $v$ is denoted by $\operatorname{RGC}(v)$. Note that for monotone TU games there is no difference between a minimal crucial and 
a real gaining coalition, i.e., $\operatorname{MCC}(v)=\operatorname{RGC}(v)$. With these generalized notions, the Public Good value for a TU game $v$ with player set 
$N$ and a player $i\in N$ is given by 
\begin{equation}
  \label{eq_pgv}
  \operatorname{PGV}_i(v)=\sum_{S\in \operatorname{MCC}(v), i\in S} v(S),\footnote{Note that the authors from \cite{holler1995public} 
  used the definition $\operatorname{PGV}_i(v)=\sum_{S\in \operatorname{RGC}(v), i\in S} v(S)$, while the authors from \cite{haradau2007potential} 
  used $\operatorname{PGV}_i(v)=\sum_{S\in \operatorname{MCC}(v), i\in S} v(S)$. As already mentioned, there is no difference for monotone TU games. 
  Also the axiomatization of the Public Good value from \cite{haradau2007potential} can be slightly adjusted by replacing the notion of minimal critical coalitions by 
  real gaining coalitions in their definition of $\pi(v,N)$ and the corresponding axiom of \textit{distributing the worths} of MCCs.}
\end{equation}     
so that $\operatorname{PGI}_i(v)=\operatorname{PGV}_i(v)$ if $v$ is a simple game. 
For the rest of the article, we will refer to $\operatorname{MCC}(v)$ as the minimal critical coalitions of $v$.

Let $\Gamma$ be a subclass of all TU games. A \emph{value} on $\Gamma$ is a function $\Psi$ that maps each game $v\in\Gamma$ to $\mathbb{R}^{|N|}$, where $N$ is the 
player set of $v$. An example of a value is the Public Good value $\operatorname{PGV}$, defined componentwise in Equation~(\ref{eq_pgv}). A \emph{potential} 
on $\Gamma$ is a function $P$ that maps each game $v\in\Gamma$ to a real number $P(v)$.

\begin{definition}
  A value $\Psi$ on $\Gamma$ \emph{admits a potential function} if there exists a potential $P\colon \Gamma\to\mathbb{R}$ such that
  \begin{equation} 
    \Psi_i(v)=P(v)-P(v_{-i})
  \end{equation}
  for all $v\in \Gamma$ and all $i\in N$, where $N$ is the player set of $v$ and $v_{-i}$ is the TU game with player set $N\backslash \{i\}$ defined by 
  $v_{-i}(S)=v(S)$ for all $\emptyset\subseteq S\subseteq N\backslash\{i\}$.
\end{definition}
Note that the subclass $\Gamma$ of TU games has to be closed with respect to taking subgames $v_{-i}$ in order to apply this definition. So, from a technical point 
of view we either have to include the game $v_{\emptyset}$ with empty player set in the set of TU games and subclasses of TU games $\Gamma$ or define $P(v_{\emptyset}):=0$ 
separately (which is the usual choice).\footnote{If we do not set $P(v_{\emptyset})=0$, then the potential of a value is only determined up to an additive constant.} As 
shown in \cite[Proposition 1]{haradau2007potential} the Public Good value $\operatorname{PGV}$ admits a potential $P$ on the 
class $\Gamma$ of (monotone) TU games, where
\begin{equation}
  P(v)=\sum_{S\in \operatorname{MCC}(v)} v(S).
\end{equation} 
Note that each minimal critical coalition $S$ in $v$ with $i\notin S$ is also a minimal critical coalition in $v_{-i}$ and vice versa. Analogously, that each real gaining 
coalition $S$ in $v$ with $i\notin S$ is also a real gaining coalition in $v_{-i}$ and vice versa. 

We say that a value $\Psi$ on $\Gamma$ \emph{distributes the sum of the worths of the minimal critical coalitions for all players} in $v$ iff
\begin{equation}
  \sum_{i\in N} \Psi_i(v)=\sum_{i\in N} \sum_{S\in\operatorname{MCC}(v), i\in S} v(S)=\sum_{S\in\operatorname{MCC}(v)} |S|\cdot v(S)
\end{equation}
for all $v\in \Gamma$, where $N$ is the player set of $v$. With this, \cite[Proposition 2]{haradau2007potential} states that the Public Good value 
$\operatorname{PGV}$ is the unique value that admits a potential and distributes the sum of the worths of the minimal critical coalitions for all players 
on the class of monotone TU games. The great advantage of an axiomatization via a potential is that this also gives an axiomatization for all subclasses 
$\Gamma'$ of TU games that are closed with respect to taking subgames $v_{-i}$. So, if we relax the condition $v(N)=1$ of a simple game, we also obtain 
an axiomatization for simple games. Note that while $v(N)=1$ it may happen that $v_{-i}(N\backslash\{i\})\neq 1$, i.e., $v_{-i}$ does not contain a 
winning coalition, which happens if player $i$ is a so-called \emph{vetoer}.

Another common property of values is linearity. To this end we note that TU games form an $\mathbb{R}$-vector space with sum $(v+v')(S):=v(S)+v'(S)$ and 
scalar multiplication $(\lambda \cdot v)(S):=\lambda\cdot v(S)$ for all TU games $v,v'$ with the same player set $N$, all $\lambda\in\mathbb{R}$, and 
all $S\subseteq N$. With this, a value $\Psi$ is called \emph{linear} if $\Psi(v+v')=\Psi(v)+\Psi(v')$ and $\Psi(\lambda \cdot v)=\lambda\cdot\Psi(v)$. 
From Equation~(\ref{eq_pgv}) we can directly conclude that the Public Good value $\operatorname{PGV}$ is linear. If only the first property, on the sum of 
two TU games holds, then one speaks of \emph{additivity}. Since the sum of two simple game (considered as TU games) does not need to be a simple game, the 
so-called \emph{transfer axiom} was introduced by Dubey \cite{dubey1975uniqueness}:
$$
  \Psi(v\wedge v')+\Psi(v\vee v')=\Psi(v)+\Psi(v'),
$$
where $\left(v\wedge v'\right)(S)=\min\{v(S),v'(S)\}$ and $\left(v\vee v'\right)(S)=\max\{v(S),v'(S)\}$ for all simple games $v,v'$ with the same 
player set $N$ and all coalitions $S\subseteq N$. Note that the definition of $\wedge$ and $\vee$ might also be applied to general TU games. In our 
context we only use $v\oplus v':=v\vee v'$ for two simple or TU games $v$, $v'$. Two simple games $v$ and $v'$ are called \emph{mergeable} if 
$S\in\operatorname{MWC}(v)$ and $S'\in\operatorname{MWC}(v')$ implies $S\not\subseteq S'$ and $S'\not\subseteq S$. 
The identity $\operatorname{PGI}_i(v\oplus v')=\operatorname{PGI}_i(v)+\operatorname{PGI}_i(v')$ for the raw Public Good index for two mergeable simple games 
was used in \cite{holler1983power} to axiomatize the normalized Public Good index.  
Similarly, for two $(j,k)$ games $v$ and $v'$ we define  
$(v\oplus v')(x)=\max\{v(x),v(x')\}$ for all $x\in J^n$, where $n$ is the number of players of $v$ and $v'$. 

\subsection{Generalizing the Public Good index to $\mathbf{(j,k)}$ simple games}
\label{sec_axiom}
The first question we have to answer is that for a suitable generalization of the concept of a minimal winning coalition in a simple game to an 
arbitrary $(j,k)$ simple game. Having the definition of minimal critical and real gaining coalitions for TU games in mind, we propose:
\begin{definition}
  Let $v$ be a $(j,k)$ simple game with player set $N=\{1,\dots,n\}$ and $J=\{0,1,\dots,j-1\}$. A vector $x\in J^n$ is called 
  \emph{minimal critical} if $v(x)>v(x')$ for all $x'\in J^n$ with $x'\le x$ and $x'\neq x$. The set of minimal critical vectors of $v$ is denoted by  
  $\operatorname{MCV}(v)$. 
\end{definition}
Note that for $j=2$ and $k=2$ each minimal critical vector $x$ corresponds to a minimal winning coalition $S=\{1\le i\le n\mid x_i=1\}$ in the corresponding 
simple game. For $j=2$ and arbitrary $k\ge 2$ we can embed a $(2,k)$ simple game $v$ as a TU game $\hat{v}$, so that the minimal critical vectors of $v$ are 
in $1$-to-$1$ correspondence with the minimal critical coalitions of $\hat{v}$. 

Let $\Gamma$ be a subclass of all $(j,k)$ simple games, where $j\ge 2$ and $k\ge 2$ are arbitrary but fixed. A \emph{value} on $\Gamma$ is a function $\Psi$ that maps 
each game $v\in\Gamma$ to $\mathbb{R}^{|N|}$, where $N$ is the player set of $v$. A \emph{potential} on $\Gamma$ is a function $P$ that maps each game $v\in\Gamma$ to $\mathbb{R}$.
\begin{definition}
  A value $\Psi$ on a subclass $\Gamma$ of $(j,k)$ simple games \emph{admits a potential function} if there exists a potential $P\colon \Gamma\to\mathbb{R}$ such that
  \begin{equation} 
    \Psi_i(v)=P(v)-P(v_{-i})
  \end{equation}
  for all $v\in \Gamma$ and all $i\in N$, where $N$ is the player set of $v$ and $v_{-i}$ is the $(j,k)$ simple game with player set $N\backslash \{i\}$ defined by 
  $v_{-i}(x)=v(y)$ for all $x\in J^{N\backslash\{i\}}$ and $y\in J^N$ with $y_i=0$ and $y_j=x_j$ for all $j\in N\backslash \{i\}$.\footnote{By $A^B$ we denote the 
  set of all mappings from $B$ to $A$ whose cardinality is $|A|^{|B|}$.} Moreover, we set $P(v_\emptyset):=0$ for a game $v_\emptyset$ with empty player set.
\end{definition}
Again, the subclass $\Gamma$ of $(j,k)$ simple games has to be closed with respect to taking subgames $v_{-i}$ in order to apply this definition. We observe that 
each minimal critical vector $x$ of $v$ with $x_i=0$ is also a minimal critical vector of $v_{-i}$ if we remove the entry for $x_i$ (so that it is a vector in 
$J^{N\backslash\{i\}}$) and vice versa. We say that a value $\Psi$ on a subclass $\Gamma$ of $(j,k)$ simple games \emph{distributes the sum of the worths of the 
minimal critical vectors for all players} in $v$ iff
\begin{equation}
  \sum_{i=1}^n \Psi_i(v)=\sum_{i=1}^n \sum_{x\in\operatorname{MCV}(v), x_i\neq 0} v(x)=\sum_{x\in\operatorname{MCV}(v)} v(x)\cdot\big\vert\left\{1\le i\le n\mid x_i\neq 0\right\}\big\vert=:\Lambda(v)
\end{equation}
for all $v\in \Gamma$, where $N=\{1,\dots, n\}$ is the player set of $v$.

\begin{theorem}
  \label{thm_axiomatization_pot}
  Let $j,k\ge 2$ be integers. Then, there exists a unique value $\Psi$ on the class $\Gamma$ of all $(j,k)$ simple games that admits a potential function and 
  distributes the sum of the worths of the minimal critical vectors for all players. We have
  \begin{equation}
    \label{eq_psi_jk_pgi}
    \Psi_i(v)=\sum_{x\in\operatorname{MCV}(v),x_i\neq 0} v(x)
  \end{equation}
  for all $v\in\Gamma$ and all $i$ in the player set $\{1,\dots,n\}$ of $v$. The potential function is given by
  \begin{equation}
    \label{eq_pot_jk_pgi}
    P(v)=\sum_{x\in\operatorname{MCV}(v)} v(x)
  \end{equation}
  for all $v\in\Gamma$.  
\end{theorem}
\begin{proof}
  First we assume that the potential is given by Equation~(\ref{eq_pot_jk_pgi}). Since $\Psi$ admits a potential function we have
  \begin{eqnarray*}
    \Psi_i(v) &=& P(v)-P(v_{-i})=\sum_{x\in\operatorname{MCV}(v)} v(x) \,-\, \sum_{x\in\operatorname{MCV}(v_{-i})} v_{-i}(x) \\ 
     &=& \sum_{x\in\operatorname{MCV}(v)} v(x) \,-\, \sum_{x\in\operatorname{MCV}(v_{-i})} v(x) \\ 
     &=& \sum_{x\in\operatorname{MCV}(v),x_i\neq 0} v(x)
  \end{eqnarray*}
  for all $v\in \Gamma$ and all $i$ in the player set of $v$, where we have used the relation between the minimal critical vectors of $v$ and those of 
  $v_{-i}$. Thus, Equation~(\ref{eq_psi_jk_pgi}) is valid. With this we have
  $$
    \sum_{i=1}^n \Psi_i(v)=\sum_{i=1}^n\sum_{x\in\operatorname{MCV}(v),x_i\neq 0} v(x) = \sum_{x\in\operatorname{MCV}(v)} v(x)\cdot\big\vert \left\{1\le i\le n\mid x_i\neq 0\right\}\big\vert=
    \Lambda(v),
  $$
  i.e., $\Psi$ distributes the sum of the worths of the minimal critical vectors for all players and so satisfies both axioms.
  
  For the other direction we assume that $\Psi$ admits a potential $\tilde{P}$, so that
  $$
    \Lambda(v)=\sum_{i=1}^n \Psi_i(v)=\sum_{i=1}^n\big(\tilde{P}(v)-\tilde{P}(v_{-i})\big)=n\cdot \tilde{P}(v)-\sum_{i=1}^n \tilde{P}(v_{-i}),
  $$
  which is equivalent to
  \begin{equation}
    \label{eq_pot_rec}
    \tilde{P}(v)=\frac{\Lambda(v)+\sum_{i=1}^n \tilde{P}(v_{-i})}{n}
  \end{equation}
  for each $v\in\Gamma$, where $N=\{1,\dots,n\}$ is the player set of $v$. For each $S\subseteq N$ we denote by $v_S$ the $(j,k)$ simple game 
  with player set $S$ defined by $v_S(x)=v(y)$ for all $x\in J^S$, where $y\in J^N$ with $y_j=x_j$ for all $j\in S$ and $y_j=0$ otherwise. E.g.\ $v_{-1}=v_{N\backslash\{i\}}$ and 
  $v_N=v$. Since $\left(v_S\right)_T=v_T$ for all $\emptyset\subseteq T\subseteq S\subseteq N$ Equation~(\ref{eq_pot_rec}) can be generalized to 
  $$
    \tilde{P}(v_S)=\frac{\Lambda(v_S)+\sum_{i\in S} \tilde{P}(v_{S\backslash \{i\}})}{|S|}
  $$ 
  for all $\{i\}\subseteq S\subseteq N$. So, starting from $\tilde{P}(v_\emptyset)=0$, we can recursively compute $\tilde{P}(v_S)$ for all $\emptyset\neq S\subseteq N$, so that 
  especially $\tilde{P}(v)=\tilde{P}(v_N)$ is uniquely defined.
\end{proof}
We call  the value $\Psi$ for $(j,k)$ simple games defined by Equation~(\ref{eq_psi_jk_pgi}) \emph{Public Good value} (for $(j,k)$ simple games). For the $(3,3)$ simple game 
$v$ from Example~\ref{ex_3_3_simple_game} the minimal critical vectors are $(1,1,2)$, $(1,2,0)$, $(2,0,1)$, $(2,1,0)$, and $(2,2,2)$, where $v(x)=1$ for 
all $x\in\operatorname{MCV}(v)\backslash\{(2,2,2)\}$ and $v((2,2,2))=2$. With this we compute
$$
  \Psi_1(v)=6,\quad \Psi_2(v)=5,\quad\text{and } \Psi_3(v)=4
$$
for the value $\Psi$ characterized in Theorem~\ref{thm_axiomatization_pot}.

We would like to remark that we also may motivate a different definition for a Public Good value for $(j,k)$ simple games. To this end we define the vector $y=x\!\downarrow\! i\in J^n$ 
for each $x\in J^n$ with $x_i\neq 0$ by $y_j=x_j$ for all $j\neq i$ and $y_i=x_i-1$. Assume that agent $i$ has strictly increasing costs in $i$ and that the rewards are 
strictly increasing in $v(x)$.\footnote{For $(2,2)$ simple games represented as simple games this means that entering a coalition comes at a certain cost while a coalition gets 
a reward iff it is a winning coalition.} As in the process of a coalition forming member by member we may imagine that starting from $x=\mathbf{0}$ the final vector $x$ 
forms step by step via the inverse operation of $\downarrow$.\footnote{More precisely, for each $x\in J^n$ with $x_i\neq j-1$ we can define the vector $y=x\!\uparrow\!i\in J^n$ by 
$y_j=x_j$ for all $j\neq i$ and $y_i=x_i+1$.} So, similarly, as one can argue that only minimal winning coalitions will be formed, we deduce that under the described model for 
every finally formed vector $x\in J^n$ with $v(x)\neq 0$ we have $x\in \operatorname{MCV}(v)$. Now what is the contribution of a player $i$ to a minimal critical 
vector $x$ with $x_i\neq 0$ to the worth $v(x)$? If the answer is $v(x)$, then we end up with the value characterized in Theorem~\ref{thm_axiomatization_pot}. However, 
if we have a look at the minimal critical vector $x=(2,2,2)$ in the $(3,3)$ simple game $v$ from Example~\ref{ex_3_3_simple_game}, then $v(1,2,2)=v(2,1,2)=v(2,2,1)=1$ may 
justify the assumption that every player contributes just a surplus of $1$ to the worth of vector $x$. Thus, we would obtain a value defined by
\begin{equation}
  \label{eq_pgi_variant}
  \Psi_i(v)=\sum_{x\in\operatorname{MCV}(v),x_i\neq 0} \big(v(x)-v(x\!\downarrow\!i)\big).
\end{equation}
Note the similarity to the Banzhaf index. For simple games the difference is that we sum over all minimal winning instead of all winning coalitions. For the $(3,3)$ simple 
game $v$ from Example~\ref{ex_3_3_simple_game} we would obtain 
$$
  \Psi_1(v)=5,\quad \Psi_2(v)=4,\quad\text{and } \Psi_3(v)=3.
$$  
We observe that there is no difference between both variants if $k=2$. And indeed, they match the variant introduced in \cite{sebastien2020public}. For all $(j,k)$ simple games 
not identically mapping to zero we define the normalized version
\begin{equation} 
  \label{eq_pgi_variant_normalized}
  \overline{\Psi}_i(v)=\frac{\Psi_i(v)}{\sum_{j=1^n} \Psi_j(v)}.
\end{equation}
Excluding the $(j,k)$ simple games $v\equiv 0$, we speak of \emph{non-trivial} $(j,k)$ simple games. Our next aim is an axiomatization for $\overline{\Psi}$. To this 
end we propose a generalization of mergeability for simple games:
\begin{definition}
  \label{def_mergeable}
  Two $(j,k)$ simple games $v$ and $v'$ with the same player set $\{1,\dots,n\}$ are \emph{mergeable} if
  \begin{enumerate}
    \item[(1)] $\operatorname{MCV}(v)\cap \operatorname{MCV}(v')=\emptyset$;
    \item[(2)] $x\in\operatorname{MCV}(v)$, $x'\in\operatorname{MCV}(v')$, $x\le x'$ \quad $\Rightarrow$\quad $v(x)<v'(x')$; and 
    \item[(3)] $x\in\operatorname{MCV}(v)$, $x'\in\operatorname{MCV}(v')$, $x\ge x'$ \quad $\Rightarrow$\quad $v(x)>v'(x')$.
  \end{enumerate}
\end{definition}
Note that (2) and (3) imply (1). Since $v(x)>0$ for all $x\in\operatorname{MCV}(v)$ the definition for $(2,2)$ simple games goes in line with the definition for simple games. 
Actually, we have $v(x)=1$ for every minimal critical vector of some $(j,2)$ simple game. If $k>2$, then we have to distinguish the critical vectors according to their output 
value $v(x)$. Next we study the relation of the minimal critical vectors of the sum of two mergeable $(j,k)$ simple games with those of their {\lq\lq}summand games{\rq\rq}.
\begin{lemma}
  \label{lemma_critical_vector_below}
  Let $v$ be a $(j,k)$ simple game with player set $\{1,\dots,n\}$. For each vector $x\in J^n$ with $v(x)>0$ there exists a vector $x'\le x$ with $v(x')=v(x)$ and $x'\in \operatorname{MCV}(v)$.
\end{lemma} 
\begin{proof}
  If $x\in \operatorname{MCV}(v)$, then the statement is true for $x'=x$. Otherwise there exists a player $1\le i\le n$ with $x_i\neq 0$ such that 
  $v(x)=v(x\!\downarrow\!i)$. If $x\!\downarrow\!i\in\operatorname{MCV}(v)$, then we can set $x'=x\!\downarrow\!i$ and are done. Otherwise we iteratively 
  apply the operator $\downarrow$ (which terminates since the number of players and output levels is finite).
\end{proof}
We remark that the minimal critical vector $x'$ does not need to be unique. To this end we may slightly adjust the $(3,3)$ simple game $v$ from Example~\ref{ex_3_3_simple_game} 
by setting $v(x)=1$ for $x=(2,2,2)$.
\begin{lemma}
  Let $v$ and $v'$ be two $(j,k)$ simple games with the same player set $\{1,\dots,n\}$ that are mergeable. Then, we have
  $$
    \operatorname{MCV}(v\oplus v')=\operatorname{MCV}(v)\cup\operatorname{MCV}(v').
  $$
\end{lemma}
\begin{proof}
  Consider $x\in\operatorname{v\oplus v'}$. Since $(v\oplus v')(x)=\max\{v(x),v'(x)\}$ we assume $(v\oplus v')(x)=v(x)$ and $v'(x)\le v(x)$ w.l.o.g. 
  If $x\notin\operatorname{MCV}(v)$, then there exists a player $1\le i\le n$ with $x_i\neq 0$ such that $v(x\!\downarrow \!i)=v(x)$. However, this 
  implies $(v\oplus v')(x\!\downarrow\!i)\ge v(x\!\downarrow\!i)=v(x)=(v\oplus v')(x)$, which is a contradiction. Thus, we have 
  $\operatorname{MCV}(v\oplus v')\subseteq \operatorname{MCV}(v)\cup\operatorname{MCV}(v')$.   
  
  Consider $x\in\operatorname{MCV}(v)$. First we show $v(x)>v'(x)$. To this end we apply Lemma~\ref{lemma_critical_vector_below} to conclude the existence 
  of a vector $x'\in J^n$ with $x'\le x$ and $v'(x')=v'(x)$. Now the stated inequality is implied by Definition~\ref{def_mergeable}.(3) and we have 
  $(v\oplus v')(x)=v(x)$. Assume $x\notin\operatorname{MCV}(v\oplus v')$ for a moment. Let $1\le i\le n$ be a player with $(v\oplus v')(x\!\downarrow\!i)=(v\oplus v')(x)$. 
  Since 
  $$
    (v\oplus v')(x\!\downarrow\!i)=\max\{v(x\!\downarrow\!i),v'(x\!\downarrow\!i)\}\le \max\{v(x\!\downarrow\!i),v'(x)\}<v(x)=(v\oplus v')(x), 
  $$    
  we obtain a contradiction. Thus, $\operatorname{MCV}(v)\subseteq\operatorname{MCV}(v\oplus v')$ and, by symmetry, also $\operatorname{MCV}(v')\subseteq\operatorname{MCV}(v\oplus v')$, 
  so that $\operatorname{MCV}(v)\cup \operatorname{MCV}(v')\subseteq\operatorname{MCV}(v\oplus v')$. 
\end{proof}
Note that $\operatorname{MCV}(v)\cap\operatorname{MCV}(v')=\emptyset$, i.e., we have the disjoint union $\operatorname{MCV}(v\oplus v')=\operatorname{MCV}(v)\uplus\operatorname{MCV}(v')$.   
  
We say that a minimal critical vector $x\in\operatorname{MCV}(v)$ is \emph{critical for player $i$ and output level $\tau$}  if 
$v(x)\ge \tau$ and $v(x\!\downarrow\!i)<\tau$. So, a given minimal critical vector $x\in \operatorname{MCV}(v)$ (with $x_i\neq 0$) is critical 
for $v(x)-v(x\!\downarrow\!i)$ output levels. Denoting the number of pairs $(x,\tau)$ such that $x\in\operatorname{MCV}(v)$ with $x_i\neq 0$ is critical 
for player $i$ with output level $\tau$ by $c_i(v)$, we have
\begin{equation} 
  c_i(v\oplus v)=c_i(v)+c_i(v')
\end{equation}  
for two mergeable $(j,k)$ simple games $v,v'$ with player set $\{1,\dots,n\}$ and $1\le i\le n$.
  
\begin{definition}
  Let $v$ be a $(j,k)$ simple game with player set $\{1,\dots,n\}$. A player $1\le i\le n$ is called a \emph{null player} if we have 
  $v(x)=v(x')$ for all $x,x'\in J^n$ with $x_j=x_j'$ for all $j\neq i$.
\end{definition}  
Note that we have $x_i=0$ for every null player $i$ and every minimal critical vector $x\in\operatorname{MCV}(v)$. The analog for simple games is that no 
null player is part of a minimal winning coalition. 
  
\begin{definition}
  \label{def_anonymous}
  Let $v$ be a $(j,k)$ simple game with player set $N:=\{1,\dots,n\}$ and $\pi\colon N\to N$ be a permutation, i.e., a bijection. The $(j,k)$ simple game 
  $\pi v$ is defined by $(\pi v)(x)=v(x')$ for all $x\in J^n$ where $x'_i=x_{\pi(i)}$ for all $1\le i\le n$.
  
  A value $\Phi$ on the class of (non-trivial) $(j,k)$ simple games is called \emph{anonymous} if for each permutation $\pi\colon N\to N$ we have 
  $\overline{\Psi}_i(\pi v)=\overline{\Psi}_{\pi(i)}(v)$, where $N$ is the player set of an arbitrary (non-trivial) $(j,k)$ simple game $v$ and $i\in N$ an arbitrary player.     
\end{definition}  

\begin{theorem}
  \label{thm_axiomatization_pgi_variant}
  The value $\overline{\Psi}$ defined in Equation~(\ref{eq_pgi_variant_normalized}) and Equation~(\ref{eq_pgi_variant}) is the unique value for non-trivial $(j,k)$ simple 
  games that satisfies the axioms:
  \begin{itemize}
    \item[(A1)] $i$ is a null player in $v$ $\quad\Rightarrow\quad$ $\overline{\Psi}_i(v)=0$.
    \item[(A2)] $\overline{\Psi}$ is efficient, i.e., $\sum_{i=1}^n \overline{\Psi}_i(v)=1$. 
    \item[(A3)] If $\operatorname{MCV}(v)=\{x\}$ for a game $v$, then $\overline{\Psi}_i(v)=\overline{\Psi}_j(v)$ for all players $i,j$ with $x_i,x_j\neq 0$.
    \item[(A4)] For all mergeable $(j,k)$ simple games $v,v'$ with player set $N$ we have
                $$
                  \overline{\Psi}_i(v\oplus v') =\frac{c(v)\cdot\overline{\Psi}_i(v) + c(v')\cdot\overline{\Psi}_i(v')}{c(v)+c(v')}
                $$
                for all $i\in N$, where $c(\tilde{v})=\sum_{j\in N} c_j(\tilde{v})$ for every non-trivial $(j,k)$ simple game $\tilde{v}$ with player set $N$.  
  \end{itemize}
\end{theorem}
\begin{proof}
  It is immediate that the value $\overline{\Psi}$ defined in Equation~(\ref{eq_pgi_variant_normalized}) and Equation~(\ref{eq_pgi_variant}) satisfies the axioms 
  (A1), (A2), and (A3). For (A4) we first note $\Psi_i(\tilde{v})=c_i(\tilde{v})$ for every $(j,k)$ simple game $\tilde{v}$ and every player $i$ in $\tilde{v}$. Using 
  the mergeability of $v$ and $v'$ we compute
  $$
    \overline{\Psi}_i(v\oplus v') =\frac{c_i(v\oplus v')}{c(v\oplus v')}  
    = \frac{c_i(v)+c_i(v')}{c(v)+c(v')}  
    = \frac{c(v)\cdot\overline{\Psi}_i(v)+c(v')\cdot\overline{\Psi}_i(v')}{c(v)+c(v')}.
  $$
  
  Conversely, given any value $\Phi$ on the class of non-trivial $(j,k)$ simple games satisfying the axioms (A1) through (A4) we proceed as follows. First we consider an 
  arbitrary non-trivial $(j,k)$ simple game $v$ with $|\operatorname{MCV}(v)|=1$ and let $x$ be the unique minimal critical vector. From (A1), (A2), and (A3) we conclude
  $$
    \Phi_i(v)=\left\{\begin{array}{rcl}
      1\,/\,|\{j\mid x_j\neq 0\}|  & & \text{if }x_i\neq 0, \\  
      0 & & \text{otherwise.}
    \end{array}\right.
  $$
  Now consider any non-trivial (j,k) simple game $\tilde{v}$ with player set $N$ and minimal critical vectors enumerated as $\operatorname{MCV}(\tilde{v})=\left\{x^1,\dots,x^m\right\}$. 
  Denoting the non-trivial $(j,k)$ simple game with unique minimal critical vector $x^h$ by $v^h$, where $1\le h\le m$, we can write
  $$
    \tilde{v}=v^1\oplus v^2\oplus\dots\oplus v^m.
  $$
  Note that the $v^h$ are sequentially mergeable in the sense that $v^{h+1}$ and $v^1\oplus\dots\oplus v^h$ are mergeable for each $h=1,2,\dots,m-1$. We can extend 
  (A4) inductively to a sum of such games to obtain for each player $i\in N$
  $$
    \Phi_i(\tilde{v})=\sum_{h=1}^m c(v^h)\Phi_i(v^h)/\sum_{h=1}^m c(v^h).
  $$
  Thus, the axioms (A1)-(A4) allow us to compute $\Phi_i(\tilde{v})$ for each non-trival $(j,k)$ simple game $\tilde{v}$ and each player $i$ of $\tilde{v}$, i.e., 
  there is at most one value satisfying axioms (A1)-(A4). So, given our first observation on $\overline{\Psi}$, we conclude $\Phi=\overline{\Psi}$. 
\end{proof}
We remark that the axioms (A1) and (A2) mimic similar axioms for simple or TU games that are used frequently in the literature. For axiom (A4) we refer to the discussion 
in \cite{holler1983power} noting that the proof of Theorem~\ref{thm_axiomatization_pgi_variant} is rather similar to the one of \cite[Section III]{holler1983power}. 
Note that for $k=2$ output levels axiom (A3) can be replaced by anonymity, see Definition~\ref{def_anonymous}. However, for $k>2$ we need some kind of stronger axiom 
in order to uniquely define the value of non-trivial $(j,k)$ simple games with a unique minimal critical vector. Of course, axiom (A3) might be considered to be too 
demanding for the cases where $x_i,x_j\neq 0$ and $x_i\neq x_j$. There is an ongoing discussion about properties that a reasonable power index or value should have, 
see e.g.\ \cite{allingham1975economic,ubt_eref62273}. We would also like to point the reader to the two axiomatizations of the Public Good index for $(j,2)$ simple 
games in \cite{sebastien2020public}, which share several axioms on the one hand and use a few different on the other hand. 

\medskip

Another approach to motivate the definition of a value for $(j,k)$ simple games is pursued in \cite{kurz2021axiomatizations} for the Shapley value. 
\begin{definition}
  \label{def_aver_tu_j_k}
  Let $v$ be an arbitrary $(j,k)$ simple game with player set $N=\{1,\dots,n\}$. The \emph{average game}, denoted by $\widetilde{v}$, associated to $v$  is defined by
  \begin{equation}
    \label{eq_aver_tu_j_k}
    \widetilde{v}(S)=\displaystyle\frac{1}{j^{n}(k-1)}\sum_{x\in J^{n}}{\left[v(\mathbf{(j-1)}_S\,,x_{-S})-v(\mathbf{0}_S\,,x_{-S})\right]}
  \end{equation}
  for all $ S\subseteq N$.
\end{definition}
For the $(3,3)$ simple game $v$ from Example~\ref{ex_3_3_simple_game} the average game $\tilde{v}$ is given by $\tilde{v}(\emptyset) = 0$, $\tilde{v}(\{1\})=\frac{1}{2}$, 
$\tilde{v}(\{2\})=\frac{5}{18}$, $\tilde{v}(\{3\})=\frac{1}{6}$, $\tilde{v}(\{1,2\})=\frac{2}{3}$, $\tilde{v}(\{1,3\})=\frac{2}{3}$, $\tilde{v}(\{2,3\})=\frac{1}{2}$, 
and $\tilde{v}(\{1,2,3\})= 1$.
Note that $\tilde{v}$ always is a TU game taking values between $0$ and $1$.

In \cite[Theorem 4.1]{kurz2021axiomatizations} it was shown that the Shapley value of a $(j,k)$ simple game $v$, as defined in e.g.\ \cite{freixas2005shapley}, equals 
the Shapley value of the TU game $\tilde{v}$. Unfortunately there is no such nice relation between the Public Good value and our analogs for $(j,k)$ simple games 
since for the $(3,3)$ simple game from Example~\ref{ex_3_3_simple_game} and the corresponding average TU game $\tilde{v}$ we have
$$
  \operatorname{PGV}_1(\tilde{v})=\frac{51}{18},\quad
  \operatorname{PGV}_2(\tilde{v})=\frac{44}{18},\quad\text{and }\operatorname{PGV}_3(\tilde{v})=\frac{42}{18}.
$$

\medskip

To sum up, we have seen that different generalizations of the Public Good value for TU games or the normalized Public Good index for simple games to the class of (non-trivial) $(j,k)$ 
simple games, including axiomatizations, are possible. As anticipated e.g.\ in \cite{freixas2012probabilistic}, a power index for simple games can admit more than one reasonable 
extension for $(j,k)$ simple games. From our personal point of view, Theorem~\ref{thm_axiomatization_pot} provides the most convincing variant. But this may be just 
a matter of taste or might depend on the application. The question of the public good properties of the proposed values is not touched at all. As done in \cite{sebastien2020public} 
for $(j,2)$ simple games, other power indices based on \textit{Riker's Size Principle} \cite[p.~32]{riker1962theory} may be treated similarly.


\end{document}